\documentclass[sigplan,nonacm, 10pt]{acmart}
\definecolor{BrevisLink}{rgb}{0.31,0.25,0.56}
\makeatletter
\@ACM@balancefalse
\makeatother
\usepackage{flushend}

\usepackage{mathpartir}
\usepackage{stmaryrd}
\newtheorem{proposition}{Proposition}

\usepackage{tikz}
\usetikzlibrary{arrows.meta,positioning,fit,backgrounds,calc}

\usepackage{algorithm}
\usepackage{algorithmic}

\usepackage{xspace}
\newcommand{\tool}{\textsc{Brevis}\xspace}
\newcommand{\W}[1]{\mathbb{W}_{#1}}
\newcommand{\bits}{\mathsf{Bits}}
\newcommand{\execp}{\mathsf{Exec}}

\mathchardef\acmmathcomma=\mathcode`,
\begingroup
\catcode`\,=\active
\gdef\acmbreakablemathcommas{%
  \mathcode`\,="8000
  \def,{\acmmathcomma\allowbreak}}
\endgroup

\let\acmtemplateincludegraphics\includegraphics
\RenewDocumentCommand{\includegraphics}{s O{} m}{%
  \begingroup
  \def\acmfigurefile{#3}%
  \ifdefstring{\acmfigurefile}{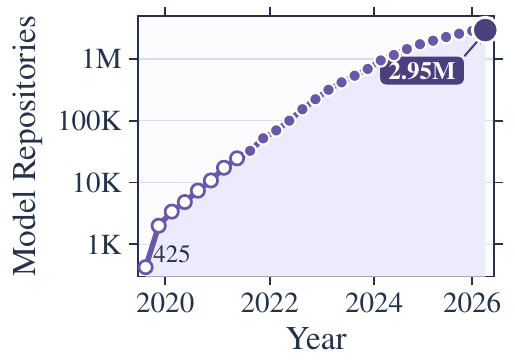}
    {\def\acmfigurefile{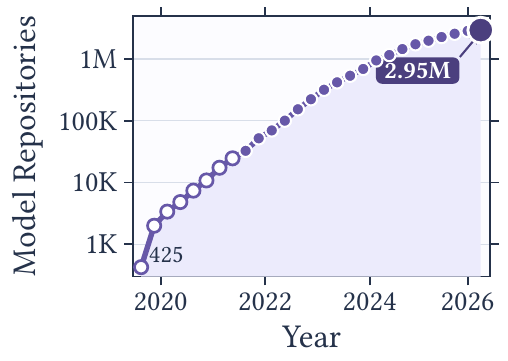}}{}%
  \ifdefstring{\acmfigurefile}{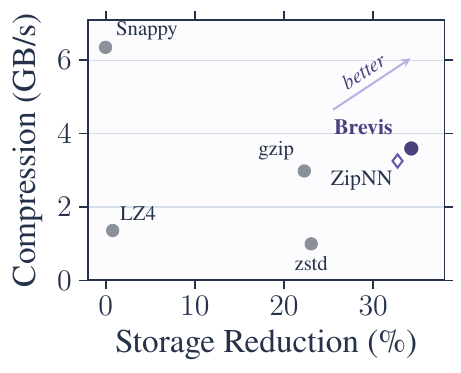}
    {\def\acmfigurefile{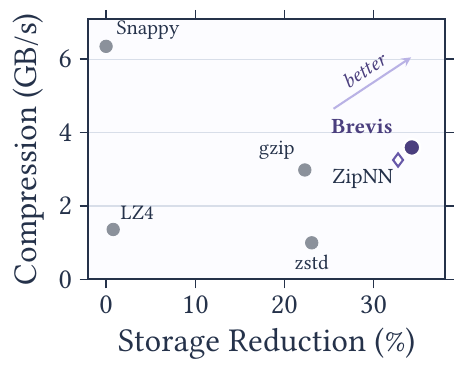}}{}%
  \ifdefstring{\acmfigurefile}{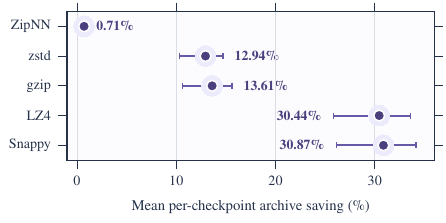}
    {\def\acmfigurefile{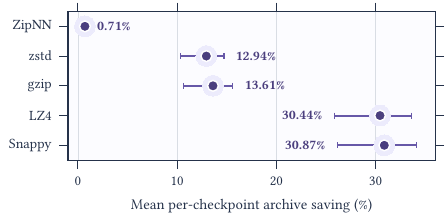}}{}%
  \ifdefstring{\acmfigurefile}{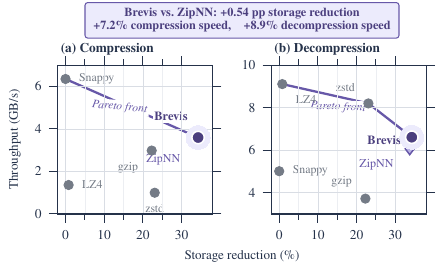}
    {\def\acmfigurefile{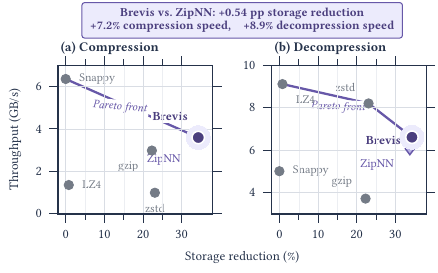}}{}%
  \IfBooleanTF{#1}
    {\acmtemplateincludegraphics*[#2]{\acmfigurefile}}
    {\acmtemplateincludegraphics[#2]{\acmfigurefile}}%
  \endgroup}

\newcommand{\acminputnumberedsection}[1]{%
  \begingroup
  \let\acmoriginalsection\section
  \RenewDocumentCommand{\section}{s o m}{%
    \IfNoValueTF{##2}
      {\acmoriginalsection{##3}}
      {\acmoriginalsection[##2]{##3}}}
  \input{#1}
  \endgroup}

\title{Lossless Tensor Compression as Program Synthesis}

\author{Jieke Shi}
\email{jiekeshi@smu.edu.sg}
\affiliation{%
  \institution{Singapore Management University}
  \country{Singapore}
}

\author{Junda He}
\email{jundahe.2022@smu.edu.sg}
\affiliation{%
  \institution{Singapore Management University}
  \country{Singapore}
}

\author{Wenjia Jiang}
\email{wjjiang@smu.edu.sg}
\affiliation{%
  \institution{Singapore Management University}
  \country{Singapore}
}

\author{Weifeng Sun}
\email{wfsun@smu.edu.sg}
\affiliation{%
  \institution{Singapore Management University}
  \country{Singapore}
}

\author{Shidong Pan}
\email{shidong.pan@anu.edu.au}
\affiliation{%
  \institution{CSIRO's}
  \country{Australia}
}

\author{Zhensu Sun}
\email{zssun@smu.edu.sg}
\affiliation{%
  \institution{Singapore Management University}
  \country{Singapore}
}

\author{Chengran Yang}
\email{cryang.2021@smu.edu.sg}
\affiliation{%
  \institution{Singapore Management University}
  \country{Singapore}
}

\author{Peixin Zhang}
\email{pxzhang@smu.edu.sg}
\affiliation{%
  \institution{AIDX TECH PTE LTD}
  \country{Singapore}
}

\author{Yifan Jia}
\email{yifan.jia@aidxtech.com}
\affiliation{%
  \institution{AIDX TECH PTE LTD}
  \country{Singapore}
}

\author{Zhou Yang}
\email{zy25@ualberta.ca}
\affiliation{%
  \institution{University of Alberta \& Alberta Machine Intelligence Institute}
  \country{Canada}
}

\author{Thong Hoang (James)}
\email{James.Hoang@csiro.au}
\affiliation{%
  \institution{CSIRO's}
  \country{Australia}
}

\author{Xiwei (Sherry) Xu}
\email{Xiwei.Xu@csiro.au}
\affiliation{%
  \institution{CSIRO's}
  \country{Australia}
}

\author{Zhenchang Xing}
\email{Zhenchang.Xing@csiro.au}
\affiliation{%
  \institution{CSIRO's}
  \country{Australia}
}

\author{David Lo}
\email{davidlo@smu.edu.sg}
\affiliation{%
  \institution{Singapore Management University}
  \country{Singapore}
}

\renewcommand{\shortauthors}{Jieke Shi et al.}

\makeatletter
\def\@mkauthors{%
  \global\setbox\mktitle@bx=\vbox{%
    \noindent\unvbox\mktitle@bx\par\medskip
    \centering
    {\large\normalfont
      Jieke Shi\textsuperscript{1},
      Junda He\textsuperscript{1},
      Wenjia Jiang\textsuperscript{1},
      Weifeng Sun\textsuperscript{1},
      Shidong Pan\textsuperscript{2},
      Zhensu Sun\textsuperscript{1},
      Chengran Yang\textsuperscript{1},\\[-0.1ex]
      Peixin Zhang\textsuperscript{3},
      Yifan Jia\textsuperscript{3},
      Zhou Yang\textsuperscript{4},
      Thong Hoang (James)\textsuperscript{2},
      Xiwei (Sherry) Xu\textsuperscript{2},
      Zhenchang Xing\textsuperscript{2},
      David Lo\textsuperscript{1}\par}
    \vspace{0.45em}
    {\small
      \textsuperscript{1}Singapore Management University, Singapore
      \quad
      \textsuperscript{2}CSIRO, Australia\\[-0.1ex]
      \textsuperscript{3}AIDX TECH PTE LTD, Singapore
      \quad
      \textsuperscript{4}University of Alberta \& Alberta Machine Intelligence Institute, Canada\par}
    \vspace{0.35em}
    {\scriptsize\ttfamily
      \{jiekeshi,jundahe.2022,wjjiang,wfsun,zssun,cryang.2021,pxzhang,davidlo\}@smu.edu.sg\\[-0.1ex]
      shidong.pan@anu.edu.au;
      \{James.Hoang,Xiwei.Xu,Zhenchang.Xing\}@csiro.au\\[-0.1ex]
      yifan.jia@aidxtech.com;
      zy25@ualberta.ca\par}
    \medskip}}
\makeatother

\AtEndPreamble{%
  \hypersetup{
    colorlinks=true,
    linkcolor=BrevisLink,
    citecolor=BrevisLink,
    urlcolor=BrevisLink,
    filecolor=BrevisLink
  }%
}

\begin{document}

\begin{abstract}
Model checkpoints are growing in both number and size, which makes archival, transfer, and deployment increasingly costly. General-purpose compressors can reduce storage requirements but ignore tensor structure, whereas existing tensor-specific compressors rely on fixed and format-specific pipelines. We present \tool, which formulates lossless tensor compression as program synthesis. We design a typed domain-specific language (DSL) that captures recurring tensor structures, such as repeated regions and floating-point fields, through a set of reversible operators. Given a tensor, \tool synthesizes a self-contained DSL program that reconstructs it bit-exactly. A checkpoint-specific production prior, learned from a small representative sample of tensors, guides a bounded A* search to synthesize compact programs, which can later be executed directly for bit-exact decompression. On 10 public checkpoints spanning language, audio, and image generation models, \tool reduces 2.13~TB of checkpoint data to 1.41~TB, a 33.93\% storage reduction. It produces archives up to 30.87\% smaller than those of four general-purpose compressors, including zstd and gzip, and smaller archives than the tensor-specific compressors ZipNN and DFloat11. Under a practical concurrency configuration, \tool achieves 3.60~GB/s compression and 6.61~GB/s decompression while preserving every source byte.
\end{abstract}

\maketitle

\section*{Introduction}

As large language models (LLMs) become increasingly capable and widely used, the number of publicly-available model repositories continues to grow rapidly. As shown in Figure~\ref{fig:intro-motivation}(a), Hugging Face has grown from 425 publicly listed model repositories in April 2020 to 2.95 million today, hosting more than 15 petabytes (PB) of data~\cite{fahlgren2026hubstats,wang2026zipllm}. Each repository may contain multiple large tensor checkpoints or revisions that are replicated across model hubs, storage systems, and deployment clusters. As model repositories continue to proliferate, the cost of archival, transfer, and deployment grows accordingly, making efficient compression increasingly important.

\begin{figure}[t!]
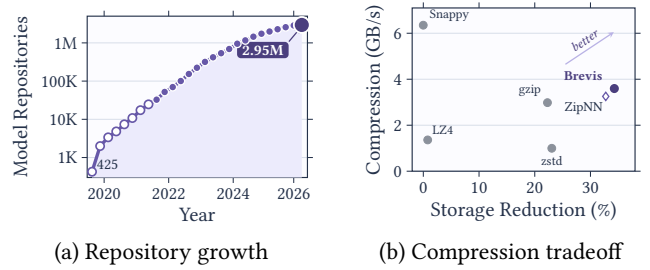

    \centering
    \begin{minipage}[t]{0.49\columnwidth}
        \centering
        \includegraphics[width=\linewidth]{figures/fig_growth_1.pdf}
        \small (a) Repository growth
    \end{minipage}\hfill
    \begin{minipage}[t]{0.45\columnwidth}
        \centering
        \includegraphics[width=\linewidth]{figures/fig_throughput_1.pdf}
        \small (b) Compression tradeoff
    \end{minipage}
    \caption{Growth in public model repositories and the compression tradeoff.
(a) Hugging Face repositories (log scale).
(b) Storage reduction versus compression throughput.}
    \label{fig:intro-motivation}
\end{figure}

Model compression can be broadly divided into lossy and lossless approaches. Lossy methods, such as quantization and pruning, reduce model size by modifying or discarding weight information~\cite{frantar2023gptq,lin2024awq,xu2025llm265}. They are designed for efficient inference and often require specific numerical formats (e.g., GGUF~\cite{githubGgmldocsggufmdMaster}) or execution environments, making them unsuitable for archival or exact checkpoint transfer because the original weights cannot be recovered bit exactly. In contrast, lossless compression preserves every source bit while reducing checkpoint size. General-purpose compressors, such as gzip and Zstandard (zstd), are widely used for storage and transmission but treat tensor weights as generic byte sequences, ignoring tensor-specific information such as dtypes, shapes, floating-point layouts, and relationships among elements. Model-specific compressors exploit common properties of floating-point tensors, such as redundancy in exponent fields~\cite{hershcovitch2025zipnn,zhang2025dfloat11}, while scientific-data compressors combine predefined transformations and codecs for floating-point arrays~\cite{burtscher2016fpcrush,claggett2018spdp,rodriguez2024adaptive}. Although effective, these methods rely on fixed compression schemes or predefined pipelines, which may fail to capture tensor-specific patterns such as repeated values, recurring subsequences, and simple relationships among elements.

This limitation raises a natural question: rather than selecting a fixed compression scheme or predefined pipeline, {\it can we synthesize a compact program that directly represents a tensor?} Such a program serves as the compressed representation and can be executed to reconstruct the original tensor bit exactly. Since each tensor can use a different program, this formulation can capture tensor-specific structures. Realizing this idea introduces three challenges. First, the language must express diverse tensor structures while guaranteeing bit-exact reconstruction. Second, the program space is large, so the search must efficiently discover compact programs within a practical budget. Third, self-contained programs must remain compact despite the overhead of storing instructions, parameters, and literals required for reconstruction.

We present \tool, which addresses these challenges with a typed domain-specific language (DSL) whose reversible operators capture recurring tensor structures, including repeated values, subsequences, element relations, and floating-point fields, while exact literals provide a fallback for values that cannot be represented more compactly. During compression, \tool learns a checkpoint-specific production prior, i.e., a probability distribution over DSL productions, from a small representative sample of tensors. This prior guides a novel bounded A* search that prioritizes promising program expansions and explores the search space within a search budget. Each complete candidate is evaluated by its exact serialized size, and the smallest synthesized program is selected for storage. During decompression, the synthesized program is executed directly for high-throughput bit-exact reconstruction, without requiring search or the learned prior.

We evaluate \tool on 10 public checkpoints spanning language, audio, and image generation models, three floating-point formats, 420 shards, and 2.13~TB of checkpoint data. \tool produces the smallest archive on all 10 checkpoints against six baselines, reducing archive size by 12.94\%-30.87\% on average relative to four general-purpose compressors. Compared with the tensor-specific baselines ZipNN and DFloat11, \tool produces archives up to 2.90\% smaller and saves 10.53~GB in total. The improvements over all paired baselines remain statistically significant after Holm correction ($p_{\mathrm{H}}=0.0098$, $r_{\mathrm{rb}}=1.0$). Figure~\ref{fig:intro-motivation}(b) shows the storage-throughput tradeoff, where \tool reaches 3.60~GB/s compression and 6.61~GB/s decompression while lying on the Pareto frontier.

This paper makes the following contributions:
\begin{itemize}
    \item We formulate bit-exact checkpoint compression as program synthesis, where each tensor is represented by a compact, self-contained DSL program that reconstructs the original tensor exactly.

    \item We design a typed domain-specific language (DSL) with reversible operators that capture recurring tensor structures and element relationships.

    \item We develop a bounded A* synthesis algorithm guided by a checkpoint-specific production prior learned from a small representative sample of tensors.

    \item We evaluate \tool on 10 public checkpoints totaling 2.13~TB of checkpoint data. \tool significantly outperforms six baselines, producing the smallest archive on every checkpoint and reducing archive size by up to 30.87\% over existing compressors.
\end{itemize}

\section{Background and Related Work}

\subsection{Lossless Compression}

Given a tensor $X$ and archive $c$, lossless compression requires
\begin{equation}
  \mathsf{Decompress}(c)=_{\mathrm{bit}}X,
  \qquad
  \mathsf{Reduction}(X,c)=1-\frac{|c|}{|X|},
\end{equation}
where $=_{\mathrm{bit}}$ denotes bitwise equality and $|c|$ includes all data and metadata required for reconstruction. General-purpose codecs, such as gzip and Zstandard, combine dictionary matching, reversible transforms, and entropy coding but treat tensors as generic byte streams~\cite{deutsch1996gzip,collet2021zstd}. Float-aware preprocessing further exploits numerical structure through Bitshuffle, Typed Data Transformation (TDT), and ALP~\cite{masui2015bitshuffle,jamalidinan2025tdt,afroozeh2023alp}. These methods improve compression through fixed transformations or predefined codec families.

Model-specific compressors exploit statistical properties of learned weights. ZipNN rearranges floating-point fields before lossless coding~\cite{hershcovitch2025zipnn}; DFloat11, ECF8, and ZipMoE target exponent distributions or model-specific layouts~\cite{zhang2025dfloat11,yang2026ecf8,yang2026zipmoe,wang2026zipllm}; Huff-LLM and tile-aligned ANS integrate entropy coding with efficient inference~\cite{yubeaton2025huffllm,tan2026shannon}; ZipServ, DFloat11, and ENEC further co-design representations with hardware decoders~\cite{fan2026zipserv,zhang2025dfloat11,yang2026enec}. In contrast, \tool targets bit-exact archival by synthesizing a self-contained tensor program rather than selecting a predefined codec or hardware-specific representation.

\begin{figure*}[t!]
  \centering
  \input{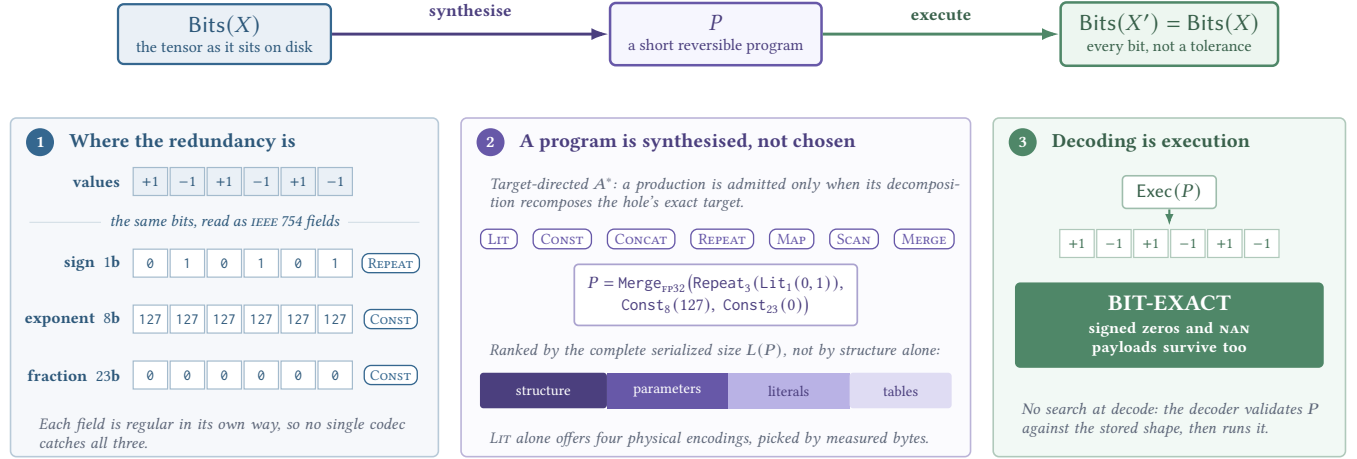}
  \caption{Compression as program synthesis in \tool. Each tensor is
  represented by a short reversible program whose execution regenerates its
  bits exactly. (1)~Read as IEEE 754 fields, 6 FP32 words expose 3
  different regularities, and no single codec captures all of them.
  (2)~Target-guided $A^*$ searches the typed grammar and ranks candidates by
  the complete serialized size $L(P)$, which includes the literal payload and
  codec tables rather than the program structure. (3)~Decoding validates
  $P$ and executes it without search and without the learned prior; a
  universal \textsc{Lit} fallback keeps every supported tensor
  representable.}
  \label{fig:overview}
\end{figure*}

\subsection{Lossy Tensor and Model Compression}

Lossy compression trades information for storage or execution efficiency. Quantization and pruning reduce model size by modifying weights~\cite{frantar2023gptq,lin2024awq}; LLM.265 repurposes video codecs for tensors~\cite{xu2025llm265}; and NeuZip studies lossless training with near-lossless inference~\cite{hao2024neuzip}. These methods target inference efficiency rather than exact reconstruction, whereas \tool focuses on bit-exact checkpoint archival.

\subsection{Program Synthesis}

Program synthesis searches a program space $\mathcal P$ for a program that satisfies a specification $\phi$ while minimizing a cost function $\mathsf{Cost}$:
\begin{equation}
P^* \in \arg\min_{P\in\mathcal P}\mathsf{Cost}(P)
\quad\text{s.t.}\quad
\phi(P,X).
\end{equation}
Syntax-guided synthesis restricts the search space with a grammar, while probability-guided methods prioritize likely productions~\cite{alur2013sygus}. PHOG, Euphony, and TF-Coder learn production probabilities or operation priors to accelerate search~\cite{bielik2016phog,lee2018euphony,shi2022tfcoder}, and the KoLMogorov Test studies exact sequence generation through synthesized programs~\cite{yoran2025kolmogorov}. Compression systems have also synthesized floating-point algorithms or searched transformation pipelines, including OpenZL's graph-based framework~\cite{burtscher2016fpcrush,claggett2018spdp,rodriguez2024adaptive,collet2025openzl}. Unlike these approaches, \tool formulates tensor compression itself as program synthesis: the DSL defines the search space, exact reconstruction defines correctness, serialized program size defines the objective, and a checkpoint-specific prior guides bounded A* search.

\begingroup
\acmbreakablemathcommas
\section{Method}

\subsection{Overview and Problem Formulation}

\tool represents each tensor as a typed, self-contained program whose execution
reconstructs the original bits. Figure~\ref{fig:overview} summarizes the
workflow. During compression, \tool learns a checkpoint-specific production
prior, i.e., a probability distribution over DSL productions, from a small
deterministic sample of tensors. This prior guides bounded A* synthesis over a
tensor-oriented grammar, while exact serialized size determines the selected
program. Reversible operators expose repeated regions, element relationships,
and low-entropy floating-point fields, while literal codecs encode the
remaining streams. During decompression, the synthesized program is validated
and executed directly without search or the learned prior.

\begin{figure*}[t]
  \centering
  \input{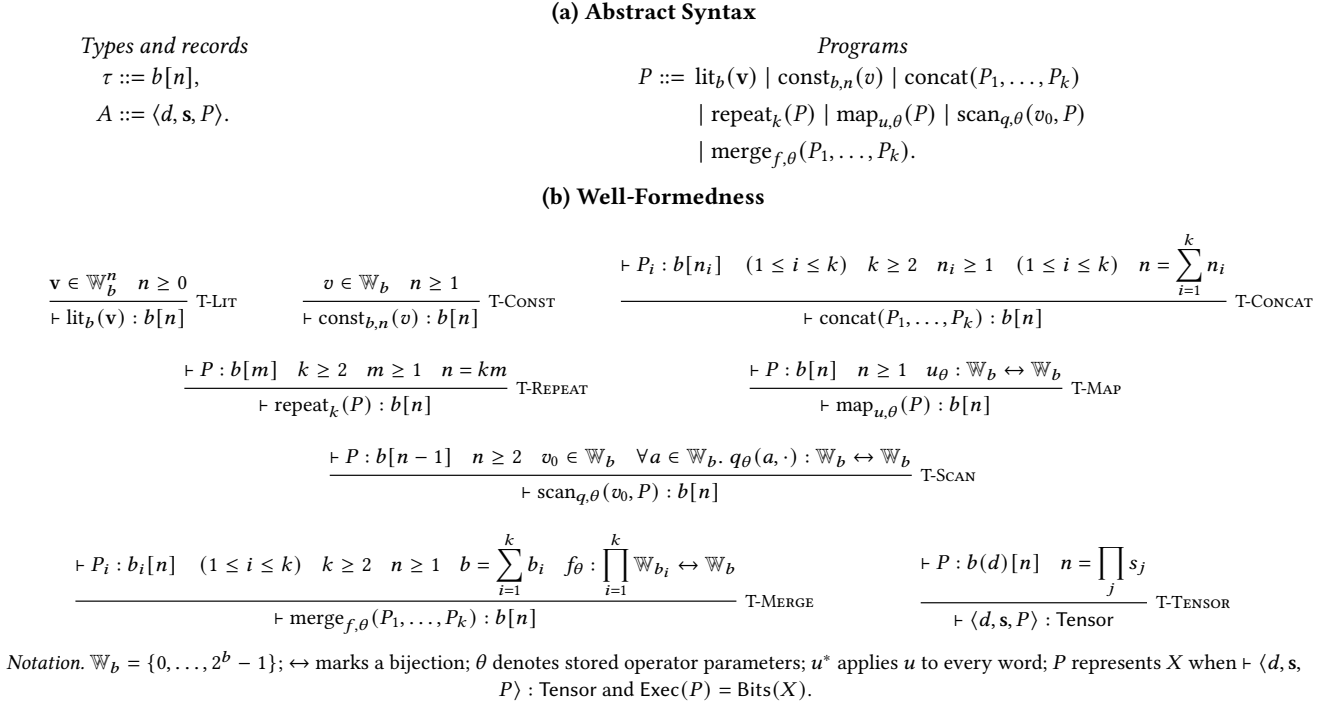}
  \caption{Typed operators in the \tool language. \textsc{Lit} is the
  universal fallback; all internal operators are reversible for their stored
  parameters.}
  \label{fig:dsl}
\end{figure*}

Let tensor $X$ have dtype $d$, shape $\mathbf{s}$, and $n$ elements. If $b(d)$
is the physical width of $d$, \tool flattens $X$ in checkpoint order into
\begin{equation}
  \mathbf{x}=\bits(X)\in\W{b(d)}^n,
  \qquad
  \W{b}=\{0,\ldots,2^b-1\}.
  \label{eq:tensor-bits}
\end{equation}
This representation preserves signed zeros, NaN payloads, and every other bit
pattern. We write $P:b[n]$ when program $P$ produces $n$ words of width $b$,
and $\execp(P)$ for its output. Program $P$ represents $X$ exactly when
\begin{equation}
  \execp(P)=\bits(X).
  \label{eq:program-correctness}
\end{equation}
The type $b[n]$ allows both the synthesizer and decoder to reject incompatible
widths and lengths.

The canonical encoding of $P$ contains its operation tags, parameters, literal
coding tables, lengths, and payloads. Let $L(P)$ denote the byte length of this
complete representation. Given expansion budget $B$, \tool selects
\begin{equation}
  \widehat P_{\mathbf{x}}
  =
  \arg\min_{P\in\mathcal C_B(\mathbf{x})} L(P)
  \quad\text{subject to}\quad
  \execp(P)=\mathbf{x},
  \label{eq:compression-objective}
\end{equation}
where $\mathcal C_B(\mathbf{x})$ contains the initial literal program and all
complete candidates found within the budget. The literal fallback makes
$\mathcal C_B(\mathbf{x})$ nonempty, while the equality constraint excludes
inexact programs. Thus, \tool minimizes exact serialized size within the
configured finite search space without claiming global optimality.

\subsection{Typed Tensor Language}

Figure~\ref{fig:dsl} presents the seven DSL operators and their well-formedness
rules. Each synthesis hole carries both a type $b[n]$ and the exact stream that
its completed subprogram must generate. Grammar rules may therefore introduce
only type-compatible children, and the decoder independently validates the
same widths, lengths, arities, and parameters.

\textsc{Lit} stores an arbitrary word stream and provides a universal
fallback, whereas \textsc{Const} stores one word and its repetition count.
\textsc{Concat} joins programs for adjacent regions, and \textsc{Repeat}
expands one nonempty child multiple times. Together, these operators capture
arbitrary values, constants, piecewise regions, and recurring subsequences.

\textsc{Map}, \textsc{Scan}, and \textsc{Merge} expose relationships that are
not visible as repeated strings. \textsc{Map} applies a width-preserving
bijection, including XOR or modular addition with a constant, ZigZag
coding~\cite{protobuf2026encoding}, Gray coding~\cite{doran2007gray}, bit reversal,
and rotation. \textsc{Scan} stores an initial word followed by XOR or
modular-addition updates, exposing adjacent correlations. \textsc{Merge}
combines equal-length children whose widths sum to the parent width, using
contiguous floating-point fields, or bit and byte planes.

The basic execution rules are
\begin{align}
  \execp(\operatorname{lit}_{b}(\mathbf v))&=\mathbf v, \notag\\
  \execp(\operatorname{const}_{b,n}(v))&=(v,\ldots,v), \notag\\
  \execp(\operatorname{concat}(P_1,\ldots,P_k))
    &=\execp(P_1)\cdots\execp(P_k), \notag\\
  \execp(\operatorname{repeat}_{k}(P))&=\execp(P)^k, \notag\\
  \execp(\operatorname{map}_{u,\theta}(P))
    &=u_\theta^*(\execp(P)).
  \label{eq:dsl-semantics}
\end{align}
For $\operatorname{scan}_{q,\theta}(v_0,P)$, execution starts with $v_0$ and
applies $q_\theta$ to the previous output and each update generated by $P$.
A \textsc{Merge} applies its stored composition pointwise. Every parameter is
serialized, every child has a fixed type, and the root must produce exactly
the words required by the tensor record.

\paragraph{Physical literal coding.}
\textsc{Lit} is a semantic leaf rather than necessarily a raw byte copy. The
encoder tries raw words, fixed-width bit packing, canonical Huffman coding,
and rANS, then selects the smallest complete encoding, including codec tags,
tables, lengths, and payloads. The decoder first restores the literal words
and then executes the surrounding program. Raw words remain available when
codec metadata would outweigh the compression benefit.

\paragraph{Running example.}
Consider the FP32 sequence
$(+1.0,-1.0,+1.0,-1.0,+1.0,-1.0)$. One synthesized program is
\begin{equation}
  \begin{aligned}
  P={}&\operatorname{merge}_{\mathrm{FP32}}\!\bigl(
    \operatorname{repeat}_{3}(\operatorname{lit}_{1}(0,1)),\\[-0.2ex]
    &\operatorname{const}_{8,6}(127),
    \operatorname{const}_{23,6}(0)\bigr).
  \end{aligned}
  \label{eq:running-example}
\end{equation}
Its three children represent the sign, exponent, and fraction fields. The sign
child stores $(0,1)$ once and repeats it three times, while the other two
generate the shared exponent and fraction. \textsc{Merge} reconstructs all six
FP32 words. A direct literal remains valid, and \tool selects the synthesized
program only if its serialized form is smaller.

\subsection{Target-Directed Synthesis}

\paragraph{Target-directed expansion.}
A naive synthesizer would enumerate programs and execute each one against the
target tensor, although most candidates would fail to reproduce it. \tool
instead searches backward from the target. Every hole is paired with the exact
stream that its completed subprogram must generate, and applying an operator
decomposes that stream into the required outputs of its children.

For operator $r$ with parameters $\theta$, let $G_{r,\theta}$ compose child
streams and $D_{r,\theta}$ decompose a target stream. \tool accepts only
decompositions satisfying
\begin{equation}
  D_{r,\theta}(\mathbf{x})=(\mathbf{x}_1,\ldots,\mathbf{x}_k)
  \Longrightarrow
  G_{r,\theta}(\mathbf{x}_1,\ldots,\mathbf{x}_k)=\mathbf{x}.
  \label{eq:decomposition-contract}
\end{equation}
Thus, if each child program generates its assigned stream, the parent is
guaranteed to generate $\mathbf{x}$. Correctness is preserved by construction
rather than checked by repeatedly executing complete candidates.

Each operator provides a target-specific inverse expansion. \textsc{Repeat}
applies only when the target consists of exact copies, \textsc{Map} inverts its
bijection, \textsc{Scan} derives updates from adjacent words, and
\textsc{Merge} splits each word into fields or planes. \textsc{Concat}
proposes boundaries derived from the target, while \textsc{Const} applies only
to constant streams. Each rule considers finitely many target-derived
parameters. Node, depth, arity, and memory limits keep the search finite, while
canonical forms remove identity and equivalent programs. The search budget is
therefore spent comparing exact representations rather than testing arbitrary
programs for correctness.

\paragraph{Checkpoint-specific production prior.}
The useful DSL productions vary across checkpoints. \tool therefore selects a
small deterministic sample stratified by dtype and tensor size, searches these
tensors with uniform production costs, and counts the productions used by the
smallest programs found. For each production, the context $\kappa$ records the
hole type, parent operator, child position, depth, dtype, size bucket, zero
fraction, distinct-value ratio, repetition ratio, and entropy of adjacent
differences.

\begin{algorithm}[t!]
  \caption{Bounded A* synthesis for one tensor.}
  \label{alg:synthesis}
  \footnotesize
  \begin{algorithmic}[1]
    \REQUIRE Target $\mathbf{x}:b[n]$, prior $\widehat q$, expansion budget $B$
    \STATE $P_{\mathrm{best}}\leftarrow\operatorname{lit}_{b}(\mathbf{x})$
    \STATE $Q\leftarrow\{H(b[n],\mathbf{x})\}$; $e\leftarrow0$
    \STATE $s_{\mathrm{cut}}\leftarrow\bot$
    \WHILE{$Q$ is not empty}
      \STATE $s\leftarrow\textsc{PopMin}(Q)$ by A* cost and byte bound
      \IF{$\operatorname{LB}_{L}(s)\ge L(P_{\mathrm{best}})$}
        \STATE \textbf{continue}
      \ENDIF
      \IF{$s$ is complete}
        \STATE $P_{\mathrm{best}}\leftarrow
          \arg\min_{P\in\{P_{\mathrm{best}},s\}}L(P)$
        \STATE \textbf{continue}
      \ENDIF
      \IF{$e=B$}
        \STATE $s_{\mathrm{cut}}\leftarrow
          \textsc{Prefer}(s_{\mathrm{cut}},s)$
        \STATE \textbf{continue}
      \ENDIF
      \STATE $e\leftarrow e+1$
      \STATE $H(\tau,\mathbf v)\leftarrow\textsc{LeftmostHole}(s)$
      \FOR{each valid target-directed expansion of $H(\tau,\mathbf v)$}
        \STATE push the resulting typed state into $Q$
      \ENDFOR
    \ENDWHILE
    \IF{$s_{\mathrm{cut}}\ne\bot$}
      \STATE $P_{\mathrm{roll}}\leftarrow
        \textsc{CompleteWithLiterals}(s_{\mathrm{cut}})$
      \STATE $P_{\mathrm{best}}\leftarrow
        \arg\min_{P\in\{P_{\mathrm{best}},P_{\mathrm{roll}}\}}L(P)$
    \ENDIF
    \RETURN $P_{\mathrm{best}}$
  \end{algorithmic}
\end{algorithm}

With rule counts $N(r,\kappa)$, admissible production set $\mathcal{R}(\kappa)$, and additive smoothing $\beta>0$, we follow PHOG and Euphony~\cite{bielik2016phog,lee2018euphony} to estimate the production prior:
\begin{equation}
  \widehat q(r\mid\kappa)
  =
  \frac{N(r,\kappa)+\beta}
       {\sum_{r'\in\mathcal R(\kappa)}N(r',\kappa)
        +\beta|\mathcal R(\kappa)|}.
  \label{eq:smoothed-rule-model}
\end{equation}
Unseen detailed contexts back off to coarser ones. Every valid production
retains nonzero probability, so the prior changes only the exploration order
and never makes a valid program unreachable. The prior is deterministic for a
fixed checkpoint and is not stored with the compressed representation.

\paragraph{Bounded A* search.}
A search state $s$ contains a partial typed program whose holes are paired with their required output streams; $H(\tau,\mathbf v)$ is a hole of type $\tau$ targeting $\mathbf v$. Applying production $r$ under context $\kappa$ incurs cost
\begin{equation}
  w(r,\kappa)=-\log_2 \widehat q(r\mid\kappa).
  \label{eq:rule-cost}
\end{equation}
States are prioritized by $g(s)+h(s)$, where $g(s)$ is the accumulated production cost and $h(s)$ is an admissible completion cost computed from a relaxed grammar. The relaxation preserves width and coarse length constraints while omitting target guards and concrete parameters, ensuring that $h(s)$ remains optimistic. Consequently, the learned prior affects only the exploration order, whereas candidate quality is determined independently by their serialized size.

A separate byte lower bound $\operatorname{LB}_{L}(s)$ accounts for fixed operation tags, parameters, and a minimum closing cost for each remaining hole. States whose lower bound cannot improve upon the current best program are pruned, while complete candidates are evaluated using their exact serialized size $L(P)$. Algorithm~\ref{alg:synthesis} summarizes the procedure. The initial incumbent is a literal program for the complete target, and search continues after the first complete candidate because the most probable program need not be the smallest. When the expansion budget is exhausted, \textsc{Prefer} retains the open state the queue orders first, and \tool completes its remaining holes with literals and evaluates the rollout, allowing structures discovered near the search boundary to remain competitive.

If the frontier is exhausted or safely pruned, the result is optimal within
the configured finite search space. If the expansion budget is reached,
\tool returns the smallest exact candidate encountered, including the
literal-completed rollout. The procedure does not claim global optimality over
unbounded programs, but every returned program reconstructs the target
exactly.

\subsection{Archive Format and Bit-Exact Reconstruction}

The archive preserves the original \texttt{safetensors} header and layout.
Each tensor record stores its name, dtype, shape, and synthesized program.
Tensor records can be compressed and decoded independently by a bounded worker
pool, while source-order emission keeps the archive deterministic. Before
execution, the decoder parses the complete program and validates the type
rules in Figure~\ref{fig:dsl}. It rejects invalid widths, lengths, parameters,
or literal payloads, executes the program, and restores the bytes at their
original offsets.

\begin{table*}[ht!]
  \caption{Complete archive results across 10 checkpoints. Each method cell
  reports archive size in decimal GB/storage reduction. Bold cells and
  underlined reductions mark the best and second-best comparable results.}
  \label{tab:checkpoint-results}
  \centering
  \scriptsize
  \renewcommand{\arraystretch}{1.08}
  \newcommand{\evalcell}[2]{#1/#2\%}
  \newcommand{\secondcell}[2]{#1/\underline{#2\%}}
  \newcommand{\bestcell}[2]{\textbf{#1/#2\%}}
  \setlength{\tabcolsep}{2.1pt}
  \begin{tabular*}{\textwidth}{@{\extracolsep{\fill}}llccccccccc@{}}
    \toprule
    Domain & Checkpoint & Data type & Source &
      \multicolumn{7}{c}{Archive size (GB)/storage reduction} \\
    \cmidrule(lr){5-11}
    & & & (GB) & zstd & LZ4 & gzip$^\dagger$ & Snappy & ZipNN &
      DFloat11$^\ddag$ & \textbf{\tool} \\
    \midrule
    Language & BERT Base & FP32 & 0.440 &
      \evalcell{0.407}{7.59} & \evalcell{0.440}{0.00} &
      \evalcell{0.407}{7.56} & \evalcell{0.440}{$-$0.01} &
      \secondcell{0.366}{16.84} & -- & \bestcell{0.365}{17.11} \\
    & Llama-3.1-8B & BF16 & 16.061 &
      \evalcell{12.38}{22.91} & \evalcell{15.94}{0.77} &
      \evalcell{12.50}{22.17} & \evalcell{16.06}{0.00} &
      \secondcell{10.66}{33.63} & \evalcell{10.90}{32.16} &
      \bestcell{10.58}{34.13} \\
    & Ministral-3-8B & BF16 & 17.836 &
      \evalcell{13.79}{22.67} & \evalcell{17.70}{0.79} &
      \evalcell{13.94}{21.85} & \evalcell{17.84}{0.01} &
      \secondcell{11.83}{33.68} & -- & \bestcell{11.74}{34.18} \\
    & Qwen3-32B & FP8 & 34.323 &
      \evalcell{28.41}{17.21} & \evalcell{34.30}{0.07} &
      \evalcell{28.29}{17.57} & \evalcell{34.33}{$-$0.01} &
      \secondcell{28.08}{18.18} & -- & \bestcell{27.85}{18.85} \\
    & Qwen3-32B & BF16 & 65.524 &
      \evalcell{50.84}{22.41} & \evalcell{65.05}{0.72} &
      \evalcell{51.36}{21.62} & \evalcell{65.53}{$-$0.01} &
      \secondcell{43.65}{33.38} & -- & \bestcell{43.28}{33.94} \\
    & Llama-3.1-70B & BF16 & 141.107 &
      \evalcell{108.58}{23.05} & \evalcell{140.00}{0.79} &
      \evalcell{109.67}{22.28} & \evalcell{141.12}{$-$0.01} &
      \secondcell{93.50}{33.74} & -- & \bestcell{92.73}{34.28} \\
    & Mixtral-8$\times$22B & BF16 & 281.241 &
      \evalcell{217.40}{22.70} & \evalcell{279.15}{0.74} &
      \evalcell{219.97}{21.79} & \evalcell{281.28}{$-$0.01} &
      \secondcell{186.23}{33.78} & -- & \bestcell{184.85}{34.27} \\
    & GLM-5.2 & BF16 & 1,506.667 &
      \evalcell{1,159.70}{23.03} & \evalcell{1,494.27}{0.82} &
      \evalcell{1,171.85}{22.22} & \evalcell{1,506.85}{$-$0.01} &
      \secondcell{998.71}{33.71} & -- & \bestcell{991.73}{34.18} \\
    \midrule
    Audio & Voxtral-Mini-3B & BF16 & 9.356 &
      \evalcell{7.215}{22.89} & \evalcell{9.269}{0.93} &
      \evalcell{7.297}{22.01} & \evalcell{9.355}{0.02} &
      \secondcell{6.226}{33.46} & -- & \bestcell{6.186}{33.89} \\
    \midrule
    Image & Qwen-Image & BF16 & 57.699 &
      \evalcell{44.64}{22.63} & \evalcell{57.25}{0.79} &
      \evalcell{45.07}{21.89} & \evalcell{57.70}{0.00} &
      \secondcell{38.42}{33.42} & -- & \bestcell{38.14}{33.89} \\
    \bottomrule
  \end{tabular*}

  \vspace{2pt}
  \parbox{\textwidth}{\scriptsize
  $^\dagger$gzip uses libdeflate 1.19 in gzip mode at DEFLATE level 1.

  $^\ddag$DFloat11 reconstructs BF16 tensor bits in its native directory format; its
  Llama-3.1-8B result was independently validated at the bit level. Dashes
  indicate unsupported or unevaluated combinations. DFloat11 is excluded from
  the paired statistical analysis.}
\end{table*}

\begin{proposition}[Compositional bit-exact reconstruction]
\label{prop:bit-exactness}
Let $P$ be a complete, well-formed program whose root is assigned
$\mathbf{x}_{\mathrm{root}}=\bits(X)$. Assume that (i) serialization and
parsing preserve the typed tree and all stored parameters, (ii) every physical
literal codec $c$ satisfies
$\operatorname{Dec}_{c}(\operatorname{Enc}_{c}(\mathbf v))=\mathbf v$, and
(iii) every internal node uses child targets returned by a decomposition
satisfying Equation~\ref{eq:decomposition-contract}. Then executing the parsed
program yields $\bits(X)$.
\end{proposition}

\begin{proof}
Let $\mathbf{x}_t$ be the target assigned to node $t$. We prove by induction
on subtree height that
\begin{equation}
  \execp(P_t)=\mathbf{x}_t
  \label{eq:node-invariant}
\end{equation}
for every parsed subtree $P_t$. For a \textsc{Lit} leaf, assumption (ii)
restores its stored stream exactly. For a \textsc{Const} leaf,
well-formedness requires the assigned target to equal the stored value repeated
to the specified length.

Now consider an internal node using production $r$, parameters $\theta$, and
children $t_1,\ldots,t_k$. By induction,
$\execp(P_{t_i})=\mathbf{x}_{t_i}$ for every child. The decoder applies the
stored composition function $G_{r,\theta}$, and
Equation~\ref{eq:decomposition-contract} gives
\begin{align}
  \execp(P_t)
    &=G_{r,\theta}\bigl(\execp(P_{t_1}),\ldots,\execp(P_{t_k})\bigr) \notag\\
    &=G_{r,\theta}(\mathbf{x}_{t_1},\ldots,\mathbf{x}_{t_k})
      =\mathbf{x}_t .
\end{align}
The invariant thus holds at the root, where
$\mathbf{x}_{\mathrm{root}}=\bits(X)$.
\end{proof}

The root type agrees with the stored dtype, shape, and layout, so placing the
generated words at their recorded offsets reconstructs $X$ bit for bit. The
guarantee assumes correct operator implementations and does not imply global
search optimality.

\endgroup
\section{Evaluation}

\subsection{Experimental Setup}

\paragraph{Corpus and baselines.}
Our corpus contains 10 public checkpoints from Hugging Face, including eight language models~\cite{devlin2019bert,grattafiori2024llama3,qwen2025qwen3,jiang2024mixtral,glm5team2026}, one audio model, and one image generation model. It spans eight BF16, one FP32, and one FP8 checkpoint, comprising 420 canonical shards and 2,130,256,127,862 source bytes (2.130~TB). We compare \tool against zstd 1.5.7 (level 9), ZipNN 0.5.4, LZ4 1.9.4 (HC level 9), gzip, and Snappy 0.7.3. We additionally compare with DFloat11 on Llama-3.1-8B, the only model for which a validated native result is available. TDT, ZipServ, and ENEC require different software or hardware backends, LLM.265 is lossy, and ECF8~\cite{yang2026ecf8} is excluded because its official CUDA validator rejected the decoded output in our experiment. These methods are therefore not directly comparable.

\paragraph{Configuration.}
The main \tool configuration uses one A* expansion per tensor and 32 workers. To learn the checkpoint-specific production prior, calibration samples at most four tensors, performs six expansions per tensor, and examines at most 1,048,576 elements from each tensor. Programs are limited to 64 nodes and depth 4, with at most 512~MiB for open decompositions. All experiments use warm file caches on one AMD EPYC 9654 server with 192 physical cores, 724~GiB RAM.

\paragraph{Measurement.}
For each checkpoint, we aggregate source and archive bytes over canonical shards and compute
\[
  \mathrm{CR}=\frac{\sum_i S_i}{\sum_i C_i},
  \qquad
  \mathrm{Saving}_m=\frac{C_m-C_{\tool}}{C_m},
\]
where $S_i$ and $C_i$ denote the source and archive sizes of shard $i$, and $C_m$ is the total archive size produced by baseline $m$. All 60 archives in the complete comparison matrix were decoded and retained only after exact reconstruction. \tool, zstd, LZ4, gzip, and Snappy reconstruct the original files byte for byte; ZipNN preserves tensor names, metadata, dtypes, shapes, and payload bits but may alter the \texttt{safetensors} layout. We compute 95\% percentile bootstrap intervals by resampling checkpoints 10,000 times with seed 20260729. For the five complete-corpus baselines, we use two-sided exact Wilcoxon signed-rank tests and report Holm-adjusted $p$-values $p_{\mathrm{H}}$ with matched-pairs rank-biserial correlation $r_{\mathrm{rb}}$.

\subsection{Compression Effectiveness}

Across the complete corpus, \tool reduces 2.130~TB to 1.407~TB, yielding a compression ratio of 1.5135 and a 33.93\% storage reduction, or 722.79~GB saved. As shown in Table~\ref{tab:checkpoint-results}, \tool produces the smallest archive for every checkpoint and available comparison. Against the four general-purpose compressors, it is smaller on all 10 checkpoints, with mean per-checkpoint savings ranging from 12.94\% over zstd to 30.87\% over Snappy (Figure~\ref{fig:baseline-improvement}). All five complete-corpus comparisons, including ZipNN, remain significant after Holm correction ($p_{\mathrm{H}}=0.0098$), with the maximum matched-pairs effect size ($r_{\mathrm{rb}}=1.0$).

ZipNN is the closest complete-corpus baseline. \tool is smaller on all 10 checkpoints and saves 10.21~GB in aggregate. Its pooled saving is 0.72\% [0.70\%, 0.81\%], and its mean per-checkpoint saving is 0.71\% [0.61\%, 0.78\%]. The mean saving across the eight BF16 checkpoints is 0.75\% [0.71\%, 0.79\%], while the FP32 and FP8 savings are 0.32\% and 0.81\%, respectively, indicating that the aggregate result is not driven solely by the largest checkpoint. We additionally compare with DFloat11 using its only publicly available and independently validated native result on Llama-3.1-8B. DFloat11 occupies 10.896~GB, compared with 10.579~GB for \tool, making \tool 316.27~MB, or 2.90\%, smaller. Because DFloat11 does not provide comparable results for the remaining checkpoints and uses a different container format, we do not generalize this comparison beyond Llama-3.1-8B.

\subsection{Throughput}
Figure~\ref{fig:throughput} compares storage reduction and throughput on Llama-3.1-70B. \tool reaches 3.60~GB/s compression and 6.61~GB/s decompression with a compression ratio of 1.522. Compared with ZipNN, it compresses 7.2\% faster while producing a 0.82\% smaller archive, placing it on both Pareto frontiers. zstd and LZ4 decompress faster, and Snappy compresses faster, but all achieve lower storage reduction. Throughput uses the 141.11-GB source size and full-checkpoint elapsed time, including process launch and output \texttt{fsync}, but excluding cache conditioning and verification. Because methods use practical rather than identical resource configurations, these comparisons are descriptive.

\begin{figure}[t]
  \centering
  \includegraphics[width=0.90\columnwidth]{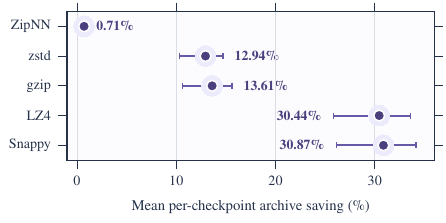}
  \caption{Mean archive saving of \tool over the five complete-corpus
  baselines. Lines show 95\% bootstrap intervals over checkpoints.}
  \label{fig:baseline-improvement}
\end{figure}

\begin{figure}[t]
  \centering
  \includegraphics[width=0.90\columnwidth]{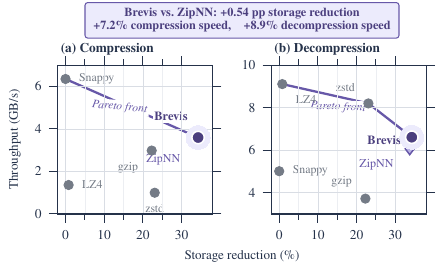}
  \caption{Storage reduction versus compression and decompression throughput. Upper right is better; lines show the Pareto frontiers.}
  \label{fig:throughput}
\end{figure}

\subsection{Synthesis Analysis and Ablation}

We analyze synthesis on the first 4.98 GB shard of Llama-3.1-8B. At budget 1, \tool already captures most of the achievable compression and completes in 9.04 seconds. Budgets 32 and 256 save an additional 3.13~MB and 5.48~MB but require 135.04 and 1,399.89 seconds, respectively, motivating the budget of one expansion used in the main evaluation. Table~\ref{tab:search-analysis} further evaluates A* search and the checkpoint-specific production prior. Removing A* increases archive size by 2,873,336 bytes, removing the prior adds 389 bytes, and removing both increases it by 3,123,197 bytes. Together, these results suggest that bounded A* and the learned prior work synergistically, with A* providing the primary compression gain while the prior further improves guided search.

\begin{table}[t]
  \caption{Search-guidance ablation on the Llama-3.1-8B
  shard. Extra bytes are relative to the full configuration.}
  \label{tab:search-analysis}
  \centering
  \begin{tabular*}{\columnwidth}{@{\extracolsep{\fill}}lrr@{}}
    \toprule
    Configuration & Extra bytes vs.\ full & Time (s) \\
    \midrule
    Full A* + prior & 0 & 140.24 \\
    No A* & +2,873,336 & 118.09 \\
    No prior & +389 & 88.82 \\
    No A* or prior & +3,123,197 & 70.69 \\
    \bottomrule
  \end{tabular*}
\end{table}

\section*{Limitations and Conclusion}

\tool formulates lossless tensor compression as program synthesis over typed, self-contained programs. Its reversible operators expose value, sequence, field, and plane structure, while literal codecs encode the resulting streams. Target-directed expansion preserves bit exactness by construction, a checkpoint-specific production prior guides bounded A* search, and exact serialized size determines the selected program. On 10 public checkpoints spanning language, audio, and image generation models, \tool reduces 2.130~TB of checkpoint data to 1.41~TB, achieving a 33.93\% storage reduction. It produces archives up to 30.87\% smaller than four general-purpose compressors, and smaller archives than the tensor-specific compressors.

Several limitations remain. The evaluation covers public model checkpoints rather than all tensor workloads, with limited and partly confounded domain and format diversity. The current implementation synthesizes each tensor independently, so cross-tensor synthesis, broader corpora, and accelerator-aware decoding remain future work.

\paragraph{Data and Artifact Availability} The anonymous implementation and datasets, as well as the scripts used to generate the evaluation results, are included in our GitHub repository: \url{https://github.com/jiekeshi/Brevis}.

\bibliographystyle{ACM-Reference-Format}
\bibliography{ref}

\end{document}